\documentclass[twocolumn,superscriptaddress]{revtex4-1}
\usepackage{amsthm}
\usepackage{amsmath}
\usepackage{latexsym}
\usepackage{amsfonts}
\usepackage{amssymb}
\usepackage{color}
\usepackage{bbm,dsfont}
\usepackage{graphicx}
\usepackage{hyperref}
\usepackage{enumerate}
\usepackage{comment}
\newtheorem{proposition}{Proposition}
\newtheorem{theorem}{Theorem}

\theoremstyle{definition}

\newcommand{\R}{\mathbb{R}} %real
\newcommand{\qq}{\mathcal{Q}} %set of POVMs
\newcommand{\bb}{\mathcal{B}} %basis
\newcommand{\no}[1]{\left\|#1\right\|} %norm
\newcommand{\tr}[1]{{\rm tr}\left[#1\right]} %trace
\newcommand{\id}{\mathbbm{1}} %identity operator
\renewcommand{\rho}{\varrho}
\newcommand{\lam}{\lambda}

\renewcommand{\Re}{{\rm Re}\,}

\begin{document}\setlength{\arraycolsep}{2pt}

\title[]{Efficient Pure State Quantum Tomography from Five Orthonormal Bases}

\author{Claudio Carmeli}
\email{claudio.carmeli@gmail.com}
\affiliation{DIME, Universit\`a di Genova, Via Magliotto 2, I-17100 Savona, Italy}

\author{Teiko Heinosaari}
\email{teiko.heinosaari@utu.fi}
\affiliation{Turku Centre for Quantum Physics, Department of Physics and Astronomy, University of Turku, FI-20014 Turku, Finland}

\author{Michael Kech}
\email{kech@ma.tum.de}
\affiliation{Department of Mathematics, Technische Universit\"at M\"unchen, 85748 Garching, Germany}

\author{Jussi Schultz}
\email{jussi.schultz@gmail.com}
\affiliation{Turku Centre for Quantum Physics, Department of Physics and Astronomy, University of Turku, FI-20014 Turku, Finland}

\author{Alessandro Toigo}
\email{alessandro.toigo@polimi.it}
\affiliation{Dipartimento di Matematica, Politecnico di Milano, Piazza Leonardo da Vinci 32, I-20133 Milano, Italy}
\affiliation{I.N.F.N., Sezione di Milano, Via Celoria 16, I-20133 Milano, Italy}

\begin{abstract}
For any finite dimensional Hilbert space, we construct explicitly five orthonormal bases such that the corresponding measurements allow for efficient tomography of an arbitrary pure quantum state. This means that such measurements can be used to distinguish an arbitrary pure state from any other state, pure or mixed, and the pure state can be reconstructed from the outcome distribution in a feasible way.  The set of measurements we construct  is independent of the unknown state, and therefore our results provide a fixed scheme for pure state tomography, as opposed to the adaptive (state dependent) scheme proposed by Goyeneche {\em et al.} in [{\em Phys. Rev. Lett.} {\bf  115}, 090401 (2015)]. We show that our scheme is robust with respect to noise in the sense that any measurement scheme which approximates these measurements well enough is equally suitable for pure state tomography. Finally, we present two convex programs which can be used to reconstruct the unknown pure state from the measurement outcome distributions. 
\end{abstract}

\maketitle

%%%%%%%%%%%%%
%\subsubsection*{Introduction}
%%%%%%%%%%%%%
{\em Introduction.---} The aim of quantum tomography is to reconstruct the unknown state of a quantum system by performing suitable measurements on it.
Tomography is a vital routine in quantum information, where it is used to characterize output states and test processing devices.
However, quantum tomography is a consuming task:  in order to obtain enough information for state reconstruction of a $d$-level system, it is necessary to perform measurements of $d+1$ different orthonormal bases, or a generalized measurement with at least $d^2$ outcomes. 
This poor scaling has led to the search for more efficient methods which allow for a reduction of resources in specific cases. 

Recent focus has been on the identification of unknown pure (or more generally low rank) states \cite{HeMaWo13, CaHeScToJPA14, CaHeScTo15, Maetal16, kech, kech2, mondragon2013determination}. Any two pure states can be distinguished with a measurement having just $\sim 4d$  outcomes \cite{HeMaWo13} or, when restricting to projective measurements, with only four orthonormal bases \cite{mondragon2013determination, Jaming14, CaHeScTo15}. The drawback of these approaches is that the measurements they provide cannot distinguish pure states from {\em all} states, implying that one needs to know that the state is pure prior to the measurement in order not to confuse it with mixed states having the same measurement outcome distributions. Moreover, none of the approaches allows an efficient recovery algorithm, mainly since the non-convex nature of the problem renders usual techniques from convex optimization useless.

%Recent focus has been on the identification of unknown pure (or more generally low rank) states \cite{HeMaWo13, CaHeScToJPA14, Goetal15, CaHeScTo15, Maetal16, kech, kech2, mondragon2013determination}. Any two pure states can be distinguished with a measurement having just $4d$  outcomes and, up to logarithmic corrections, this number is indeed optimal \cite{HeMaWo13}. 
%This amounts to a dramatic reduction for higher dimensional systems such as those consisting of multiple qubits. When restricting to projective measurements, the optimal number of measurements (i.e. orthonormal bases) is four, also giving rise to $4d$ outcomes in this constrained scenario \cite{mondragon2013determination, Jaming14, CaHeScTo15}. The drawbacks of these approaches come from the fact that the measurements they provide cannot distinguish pure states from {\em all} states, meaning that it is necessary to know that the state is pure prior to the measurement, as there may be mixed states resulting in the same measurement outcome distributions. Moreover, none of the approaches provides an efficient recovery algorithm, mainly since the non-convex nature of the problem renders natural approaches from convex optimization useless.
%

In \cite{Goetal15}, a scheme involving five orthonormal bases along with a reconstruction algorithm was proposed and experimentally demonstrated. Remarkably, such a scheme allows to certify the purity assumption on the state directly from the measurement outcomes. However, this method is adaptive in the sense that the outcome distribution of the first measurement affects the choice of the subsequent ones. 
Therefore, if one requires the procedure to work for all pure states the overall number of required measurement settings is considerably larger than five.

At the cost of a slightly higher number $\mathcal{O}(d\ln d)$ of measurement outcomes, tomographic procedures based on compressed sensing allowing for the stable recovery of pure quantum states were proposed in \cite{gross2010quantum, QCS2, QCS3}.
Rather than providing a functioning measurement setup, these results guarantee that, with high probability, any state can be reconstructed by using sufficiently many randomly drawn measurement settings. From a practical point of view, however, a deterministic approach which provides an explicit measurement set-up may be favourable.

In this Letter we overcome these drawbacks by constructing five orthonormal bases such that every pure state can be efficiently reconstructed from the corresponding measurements. For any dimension $d$, our set of measurements is fixed and therefore there is no need for data processing in between the measurements. We show that these measurements distinguish pure states from all states, and this therefore shows that the scaling $\sim 5d$  in the total number of outcomes is the same as without the constraint of having projective measurements  \cite{CDJJKSZ88}.
More importantly, we prove that the presented set-up is robust with respect to noise.
Finally, we provide reconstruction algorithms for the practical retrieval of the unknown state from the measurement data. We remark that, as compared to the compressed sensing results of \cite{gross2010quantum, QCS2, QCS3}, our result comes with fewer measurement outcomes. However, the stability guarantees that we can derive are weaker. 

%It should be stressed that our results provide, for any dimension $d$, a fixed set of measurements as opposed to the adaptive scheme presented in \cite{Goetal15}, where the outcome distribution of the first measurement affects the choice of the subsequent ones. In other words, our method can be clearly divided into the measurement part and the reconstruction part, without the need for any data processing in between.  We remark that, as compared to the compressed sensing results of \cite{gross2010quantum, QCS2, QCS3}, our result comes with fewer measurement outcomes. However, the stability guarantees we can derive are weaker. 
%

%Evidently, if the unknown state is not completely arbitrary but we have some prior knowledge or assumption about it, then we could hope for a less demanding task.
%In this paper we concentrate on the following questions:
%How many different measurement settings are needed in order to determine an unknown pure state among all states? How should such measurements be chosen? Is this choice of the 

%We will consider a finite $d$-dimensional quantum system.
%By $\mathcal{S}^d$ we denote the set of quantum states, and by $\mathcal{S}^d_1$ the subset of pure quantum states.

%%%%%%%%%%%%%
%\subsubsection*{Separating Pure States From All States}
%%%%%%%%%%%%%
%{\em Separating pure states from all states.---}
{\em Construction of the bases.---}  We begin by constructing, for any dimension $d$, five orthonormal bases $\bb^0,\ldots, \bb^4$ which determine any pure state among all states.
This means that for any pure state represented by a unit vector $\psi$, and any density matrix $\varrho$, the equalities
$$
\left\vert \langle v^\ell_j \mid \psi\rangle \right\vert^2 = \langle v^\ell_j \mid \varrho  \, v^\ell_j \rangle \qquad \text{ for all } v^\ell_j\in \bb^\ell
$$ 
imply that $\varrho = \vert \psi\rangle \langle \psi \vert$. The construction is an extension of \cite{Jaming14} where, based on the properties of Hermite polynomials, four orthonormal bases $\bb^1,\ldots, \bb^4$ capable of distinguishing any two pure states were presented. That construction generalizes easily to any sequence of orthogonal polynomials as explained in \cite{CaHeScTo15}. Remarkably, by adding the canonical basis $\bb^0=\{e_0, \ldots, e_{d-1}\}$ to this set, we obtain the five bases  with the desired property. 

%Based on Hermite polynomials, in \cite{Jaming14} four orthonormal bases were constructed which can distinguish any two pure states. 
%Emanating form this approach, in the following we use sequences of orthogonal polynomials to construct five orthonormal bases which can determine an arbitrary pure state among all states. 

To begin with, let us fix a  \emph{sequence of orthogonal polynomials}, that is, a sequence $(p_n)_{n=0}^{\infty}$ of real polynomials such that $p_n$ is of degree $n$ and 
\begin{align*}
\langle p_j,p_i\rangle:=\int_{-\infty}^{\infty} p_j(x)p_i(x)w(x)dx=\delta_{ij}
\end{align*}
for some positive weight function $w$. For a $d$-dimensional system we will only need the first $d+1$ polynomials. To construct the first two bases, let $x_0,\hdots,x_{d-1}$ be the zeros of $p_d$, which are real and distinct numbers satisfying $p_{d-1}(x_j) \neq 0$ for all $j\in\{0,\hdots,d-1\}$ \cite[Section 3.3]{Szego}. Pick an $\alpha\in\mathbb{R}$ such that $e^{i j\alpha}\notin\mathbb{R}$ for all $j\in\{1,\hdots,d-1\}$. Now for $j=0,\hdots,d-1$, set
%\begin{align}\label{eq1}
%&v^1_j:=\left(p_0(x_j),p_1(x_j),\hdots,p_{d-1}(x_j)\right),\\
%\label{eq2}
%&v^2_j:=\left(p_0(x_j),e^{i\alpha}p_1(x_j),\hdots,e^{i(d-1)\alpha}p_{d-1}(x_j)\right).
%\end{align}
\begin{align*}
&v^1_j:=\left(p_0(x_j),p_1(x_j),\hdots,p_{d-1}(x_j)\right),\\
&v^2_j:=\left(p_0(x_j),e^{i\alpha}p_1(x_j),\hdots,e^{i(d-1)\alpha}p_{d-1}(x_j)\right)
\end{align*}
and denote $\bb^1 = \{ v^1_j /\Vert v^1_j\Vert\mid j=0,\ldots, d-1 \}$ and $\bb^2 = \{ v^2_j /\Vert v^2_j\Vert \mid j=0,\ldots, d-1\}$. The fact that these are actually orthonormal bases can be readily checked using the Christoffel-Darboux formula \cite[Theorem 3.2.2]{Szego}
\begin{align*}
\sum_{i=0}^n p_i(x)p_i(y)=\frac{k_n}{k_{n+1}} \frac{p_{n+1}(x)p_n(y)-p_n(x)p_{n+1}(y)}{x-y}
\end{align*}
where $k_n$ is the leading coefficient of $p_n$ (see \cite{Jaming14, CaHeScTo15} for more details). This formula evaluated at $n=d-1$ and $x=y=x_j$ also yields the normalization factor
$$
\no{v^1_j}^2 = \no{v^2_j}^2 = \frac{k_{d-1}}{k_d}\, p'_d(x_j)p_{d-1}(x_j) .
$$
For the remaining two bases, let $y_0,\hdots,y_{d-2}$ be the zeros of $p_{d-1}$. As the polynomials $p_d$ and $p_{d-1}$ have no common zeros, the $y_j$:s are distinct from the $x_j$:s. By a similar reason, $p_{d-2}(y_j)\neq 0$ for all $j=0,\hdots,d-2$. For $j=0,\hdots,d-2$ define the non zero vectors
%\begin{align}\label{eq3}
%&v^3_j:=\left(p_0(y_j),p_1(y_j),\hdots,p_{d-2}(y_j),0\right),\\
%\label{eq4}&v^4_j:=\left(p_0(y_j),e^{i\alpha}p_1(y_j),\hdots,e^{i(d-2)}p_{d-2}(y_j),0\right) \, , 
%\end{align}
\begin{align*}
&v^3_j:=\left(p_0(y_j),p_1(y_j),\hdots,p_{d-2}(y_j),0\right),\\
&v^4_j:=\left(p_0(y_j),e^{i\alpha}p_1(y_j),\hdots,e^{i(d-2)\alpha}p_{d-2}(y_j),0\right) \, , 
\end{align*}
and by setting $v^3_{d-1}:=e_{d-1}$ as well as $v^4_{d-1}:=e_{d-1}$ we have arrived at the two orthonormal bases $\bb^3 = \{ v^3_j /\Vert v^3_j\Vert\mid j=0,\ldots, d-1 \}$ and $\bb^4 = \{ v^4_j/\Vert v^4_j\Vert \mid j=0,\ldots, d-1\}$. The normalization is now given by
$$
\no{v^3_j}^2 = \no{v^4_j}^2 = \frac{k_{d-2}}{k_{d-1}}\, p'_{d-1}(y_j)p_{d-2}(y_j) .
$$

\begin{theorem}\label{thmpca}
The five orthonormal bases $\bb^0,\ldots, \bb^4$ constructed above determine any pure state among all states.
\end{theorem}
\begin{proof}
Let $\psi = \sum_{j=0}^{d-1} c_j e_j$ be a unit vector  and let $\varrho$ be an arbitrary state such that  $\left\vert \langle v^\ell_j \mid \psi\rangle \right\vert^2 = \langle v^\ell_j \mid \varrho \, v^\ell_j \rangle$ for all $v^\ell_j\in\bb^\ell$. 
%imply inductively that the matrix element of the two states coincide, i.e., that $
%\varrho_{k,l} = \langle e_k\vert \varrho \, e_l\rangle = \langle e_k\vert \psi\rangle\langle \psi \vert e_l\rangle = c_k\overline{c_l} $.
From the standard basis $\bb^0$ we get $\varrho_{k,k} =\vert c_k\vert^2$ for all $k$. Let $n$ denote the largest number such that $\varrho_{n,n}=\vert c_n\vert^2\neq 0$ so that  by the positivity of $\varrho$, $\varrho_{k,l} = \varrho_{l,k} = 0$ for all $k > n$.  By the definition of the bases and the equalities of the probabilities we then have
%\begin{align*}
%\sum_{k,l=0}^{n} \rho_{k,l} p_k (z) p_l (z) &= \sum_{k,l=0}^{n} c_k \overline{c}_l p_k (z) p_l (z) \\
%\sum_{k,l=0}^{n} \rho_{k,l} e^{i(l-k)\alpha}p_k (z) p_l (z) &= \sum_{k,l=0}^{n} c_k \overline{c}_l e^{i(l-k)\alpha}p_k (z) p_l (z)
%\end{align*}
\begin{align}
&\sum_{k,l=0}^{n}( \rho_{k,l}  -c_k \overline{c}_l )p_k (z) p_l (z) =0\label{eqn:pol1}\\
&\sum_{k,l=0}^{n} (\rho_{k,l} - c_k \overline{c}_l) e^{i(l-k)\alpha}p_k (z) p_l (z)=0\label{eqn:pol2}
\end{align}
for all $z=x_j$ and $z=y_j$, but since the polynomials have degree at most  $2n\leq 2d-2$ and they vanish on  $2d-1$ distinct points, they must be identically zero. In other words,  the above equalities must hold for all $z\in\R$. Let us denote $t_{k,l} =  \rho_{k,l}  -c_k \overline{c}_l $ so that $t_{l,k} = \overline{t_{k,l}}$ and $t_{k,k}=0$. By looking at  the highest degree terms in  \eqref{eqn:pol1} and \eqref{eqn:pol2} we get $\Re (t_{n,n-1}) = \Re(e^{-i\alpha}t_{n,n-1})=0$, which imply that $t_{n,n-1}=0$. In other words, the matrix elements of the two states coincide on the diagonal and the bottom right $(d-n+1)\times (d-n+1)$-block. We now proceed by induction.

Firstly,  whenever the two states coincide on some bottom right $(d-r)\times (d-r)$-block, with $1\leq r\leq n-1$, we have $t_{k,l}=0$ for $k\geq r$ and $l\geq r$. But then the highest degree terms in \eqref{eqn:pol1} and \eqref{eqn:pol2} give $\Re (t_{n,r-1}) = \Re(e^{i(r-n-1)\alpha}t_{n,r-1})=0$, which yield $t_{n,r-1}=0$, that is $\varrho_{n,r-1}= c_n \overline{c_{r-1}}$. Secondly, using this and the positivity of $\varrho$ we can calculate for all $r-1 < k < n$
\begin{align*}
0 &\leq \left\vert  \begin{array}{ccc}
\rho_{r-1,r-1} & \rho_{r-1,k} & \rho_{r-1,n}\\
\rho_{k,r-1} & \rho_{k,k} & \rho_{k,n} \\
\rho_{n,r-1} & \rho_{n,k} & \rho_{n,n} 
\end{array} \right\vert
 = 
 \left\vert  \begin{array}{ccc}
\vert c_{r-1}\vert^2 & \rho_{r-1,k} & c_{r-1}\overline{c_{n}} \\
\overline{\rho_{r-1,k}} & \vert c_k\vert^2 & c_k\overline{c_{n}} \\
\overline{c_{r-1}}c_{n} & \overline{c_k}c_{n} & \vert c_{n}\vert^2 
\end{array} \right\vert \\
& = -\vert c_{n}\vert^2 \vert \rho_{r-1,k} - c_{r-1}\overline{c_k}\vert^2
\end{align*}
which, since $c_n\neq 0$,  gives us $\rho_{r-1,k} = c_{r-1}\overline{c_k}$. The two states therefore coincide on a larger  bottom right block. By induction, the states must be equal.

\end{proof}

%\begin{proof}
%The details of the proof are relegated to the Supplemental Material.
%\end{proof}

To give an example of the previously explained construction of five bases, we take the Chebyshev polynomials of the second kind $(U_n)_{n=0}^{\infty}$.
These are the unique polynomials such that \cite[p.~3]{Szego}
\begin{align*}
U_n(\cos\theta)=\frac{\sin\left((n+1)\theta\right)}{\sin\theta}
\end{align*}
holds for all $n=0,1,\ldots$ and $\theta\in [0,2\pi)$. 
The $n$ roots of $U_n$ are given by
\begin{align*}
\cos\left(\frac{j+1}{n+1}\,\pi\right),\ j=0,\hdots,n-1 \, , 
\end{align*}
and its leading coefficient is $k_n = 2^n$. 
Hence, the normalized vectors of the first basis are given by
$$
v_j^1=\sqrt{\frac{2}{d+1}}\left(\sin\left(1\,\frac{j+1}{d+1}\,\pi\right),\hdots,\sin\left(d\,\frac{j+1}{d+1}\,\pi\right)\right) \, , 
$$
and the other bases have similar and equally simple forms.

%The polynomials $(U_n)_{n\in\mathbb{N}_0}$ are orthogonal with respect to the inner product \cite{???}
%\begin{align*}
%\langle U_i,U_j\rangle_U=\int_{-1}^{1}U_i(x)U_j(x)\sqrt{1-x^2}dx=\frac{\pi}{2}\delta_{ij}
%\end{align*}
%and hence, choosing $\alpha=\frac{\pi}{2d}$, we find for $j=0,\hdots,d-1$ that
%%\small{
%\begin{align*}
%v_j^1&=\sqrt{\frac{2}{\pi}}\left(1,\frac{\sin\left(2\frac{j\pi}{d+1}\right)}{\sin\left(\frac{j\pi}{d+1}\right)},\hdots,\frac{\sin\left(d\frac{j\pi}{d+1}\right)}{\sin\left(\frac{j\pi}{d+1}\right)}\right),\\
%v_j^2&=\sqrt{\frac{2}{\pi}}\left(1,e^{i\frac{\pi}{2d}}\frac{\sin\left(2\frac{j\pi}{d+1}\right)}{\sin\left(\frac{j\pi}{d+1}\right)},\hdots,e^{i\frac{(d-1)\pi}{2d}}\frac{\sin\left(d\frac{j\pi}{d+1}\right)}{\sin\left(\frac{j\pi}{d+1}\right)}\right).
%\end{align*}
%%}
%Similarly, for $j=0,\hdots,d-2$ we find
%%\small{
%\begin{align*}
%v_j^3&=\sqrt{\frac{2}{\pi}}\left(1,\frac{\sin\left(2\frac{j\pi}{d}\right)}{\sin\left(\frac{j\pi}{d}\right)},\hdots,\frac{\sin\left((d-1)\frac{j\pi}{d}\right)}{\sin\left(\frac{j\pi}{d}\right)},0\right),\\
%v_j^4&=\sqrt{\frac{2}{\pi}}\left(1,e^{i\frac{\pi}{2d}}\frac{\sin\left(2\frac{j\pi}{d}\right)}{\sin\left(\frac{j\pi}{d}\right)},\hdots,e^{i\frac{(d-1)\pi}{2d}}\frac{\sin\left((d-1)\frac{j\pi}{d}\right)}{\sin\left(\frac{j\pi}{d}\right)},0\right).
%\end{align*}
%%}

%%%%%%%%%%%%%%%%%%%%%%
%\subsubsection*{More general measurements}
%%%%%%%%%%%%%%%%%%%%%

{\em More general measurements and stability.---} A realistic measurement is affected by noise and therefore cannot be described simply by an orthonormal basis.
Even more, an optimal measurement for a given task might not even be related to an orthonormal basis.
For these reasons, one needs to have a wider mathematical framework for measurements.
A general measurement in quantum mechanics can be modelled by a positive operator valued measure (POVM) \cite{MLQT12}, which  is a function $j\mapsto P(j)$ from a finite set of measurement outcomes $\{1,\ldots,m\}$ to the linear space of $d\times d$ Hermitian matrices $H(d)$ such that $P(j)\geq 0$ and $\sum_{j=1}^m P(j)=\id$. In practice one might want to measure more than one POVM. 
For instance, a noisy measurement of each orthonormal basis can be described by a separate POVM. 
By a {\em measurement scheme} we mean a set $\qq:=\{P_1,\hdots,P_l\}$ of POVMs.
It is not restrictive to assume that all POVMs in a given measurement scheme have the same set of outcomes $\{1,\ldots,m\}$. 
A measurement scheme $\qq$ therefore induces a linear map $M_\qq$ from the real vector space $H(d)$ to the set of real $l\times m$ matrices  $M_{lm}(\R)$ via 
\begin{align*}
M_\qq (X)_{i,j} = \tr{X P_i(j)}.
\end{align*}
The image of a state $\varrho$ is the real matrix whose $i$-th row contains the outcome probabilities corresponding to $P_i$.  Analogously to the case of projective measurements, we say that the measurement scheme $\qq$ determines any pure state among all states if for any pure state $\sigma=\vert \psi\rangle\langle \psi \vert $ and any state $\varrho$, the equality $M_\qq(\sigma) = M_\qq (\varrho)$ implies $\varrho = \sigma$.  By adapting the argument of \cite[Theorem 1]{CaHeScToJPA14}, it is easy to see that  a measurement scheme $\qq$ has this property if and only if every non-zero element of $\ker M_\qq$ (the kernel of the map $M_\qq$) has at least two positive eigenvalues (see the Supplemental Material for the detailed proof).

%
%The following formulation of the pure state determination problem for POVMs allows us to investigate the noise robustness of the problem.
%
%Let us begin by fixing some notation. The real vector space of $d\times d$ Hermitian matrices is denoted by $H(d)$, and we equip it   By $\text{ker}\,M$ and $\text{ran}\,M$ we denote kernel and range of a linear map $M$, respectively.
%
%For the purpose of the present article, a POVM is a function $j\mapsto P(j)$ from a finite set of measurement outcomes $\{1,\ldots,m\}$ to the set of Hermitian matrices $H(d)$ such that $P(j)\geq 0$ and $\sum_{j=1}^m P(j)=\id$. 
%A POVM $P$ maps a quantum state $\varrho\in\mathcal{S}^d$ to the probability distribution $M_P(\varrho)$ on the outcome set $\{1,\ldots,m\}$. 
%In practice one might want to measure more than one measurement set-up. 
%By a {\em measurement scheme} we mean a multiset $\qq:=\{P_1,\hdots,P_l\}$ of POVMs.
%It is not restrictive to assume that all POVMs in a given measurement scheme have the same set of outcomes $\{1,\ldots,m\}$. 
%A measurement scheme $\qq$ induces a linear map given by
%\begin{align*}
%M_\qq : H(d) \to M_{lm}(\R) \quad X\mapsto (M_{P_1}(X)^t,\hdots,M_{P_l}(X)^t)^t
%\end{align*}
%which maps a state $\varrho\in\mathcal{S}^d$ to the real $l\times m$ matrix having the propability $M_{P_i}(\varrho)$ as its $i$-th line. 
%

With this framework of measurement schemes we are now prepared to discuss the noise robustness of the result stated in Theorem \ref{thmpca}. 
First, we will need to have a notion of closeness of two measurement schemes, and for this reason we fix norms on the real vector spaces $H(d)$ and $M_{lm}(\R)$.
Since these are finite dimensional vector spaces, all norms are equivalent and the choice is not important for our purposes.
Typical choices are, e.g., the trace norm $\no{X}=\tr{|X|}$ on $H(d)$, and on $M_{lm}(\R)$ the supremum of the $\ell^1$-norm over all lines, i.e., 
\begin{align*}
\no{M} = \sup_{i} \sum_j |M_{i,j}| \, . 
\end{align*}
The inequality $\no{M_\qq(\varrho) - M_{\qq'}(\varrho)}\leq\epsilon$ then means that the measurement outcome distributions of all the POVMs in $\qq$ and $\qq'$ measured on the same state $\varrho$ are uniformly close in the total variation norm.
We will say that two measurement schemes $\qq$ and $\qq'$ are {\em $\epsilon$-close} if $\no{M_\qq-M_{\qq'}}_\infty < \epsilon$, where $\no{\cdot}_\infty$ is the uniform operator norm in the chosen norms of $H(d)$ and $M_{lm}(\R)$.

%%%%%%%%%%%%%%%%%%%%%%%%%%%%%%
%\section{Seperating Pure States From All States}
%%%%%%%%%%%%%%%%%%%%%%%%%%%%%%

%The central objects we study in the following are measurement schemes that can determine an arbitrary pure state from any other state. 
%More precisely, this is captured by the following definition first introduced in \cite{chen2013uniqueness}.
%
%\begin{definition}\label{def1}
%A measurement scheme $\qq$ is called {\em pure state informationally complete among all states} (PCA) if $M_\qq(\sigma)\neq M_\qq(\varrho)$ for all $\sigma\in\mathcal{S}^d_1$ and $\varrho\in\mathcal{S}^d$ with $\sigma\neq \varrho$.
%\end{definition}
%
%We aim at utilizing PCA measurement schemes for pure state quantum tomography. With this task in mind, it is crucial that our approach is robust with respect to noise. This robustness is guaranteed by the fact that Definition \ref{def1} is stable with respect to small perturbations, as proved in the next theorem.

%The following theorem now states the main result regarding the noise robustness of measurement schemes in the context of pure state tomography.  

\begin{theorem}[Stability]\label{propstab}
If a measurement scheme $\qq$ determines any pure state among all states, then there is an $\epsilon>0$ such that every measurement scheme $\qq'$ which is $\epsilon$-close to $\qq$ has this same property. 
\end{theorem}
\begin{proof}
For $i\in\{1,\hdots,d\}$, denote by $\lambda_i(X)$ the $i$-th greatest eigenvalue of a Hermitian matrix $X\in H(d)$. Let $K:=\{X\in H(d):\lambda_2(X)\leq 0, \no{X}=1\}$ be the set of unit norm Hermitian matrices with at most one positive eigenvalue. Consider the map $\phi:H(d)\to\mathbb{R}^d$,
\begin{align*}
\phi(X) =  (\lambda_1(X),-\lambda_2(X),\hdots,-\lambda_d(X))
\end{align*}
and let $L:=[0,+\infty)^d$. We have $K=\phi^{-1}(L)\cap H(d)_1$, where $H(d)_1$ is the unit sphere in $H(d)$. Since $H(d)_1$ is compact, $L$ is closed, and $\phi$ is continuous by Weyl's perturbation theorem \cite[Corollary III.2.6]{Bhatia}, we conclude that $K$ is a compact set.\\
We claim that a measurement scheme $\qq$ determines pure states among all states if and only if $c:=\min_{X\in K}\no{M_\qq(X)}>0$. 
First, assume $c>0$, and let $X\neq 0$ be such that $\lam_2(X)\leq 0$. We have $X/\no{X} \in K$, hence
$$
\no{M_\qq(X)} = \no{X}\no{M_\qq\left(\frac{X}{\no{X}}\right)} \geq c\no{X} \neq 0 ,
$$
that is, $X\notin\ker M_\qq$. Therefore, $\qq$ determines any pure state among all states. Conversely, suppose that  $\qq$ has the latter property. By the compactness of $K$, there is $X\in K$ such that $c=\|M_\qq(X)\|$. Since every non-zero element of $\ker M_\qq$ has at least two positive eigenvalues, we have  $M_\qq(X)\neq 0$ and thus $c\neq 0$.

Finally, if $\no{M_\qq-M_{\qq'}}_\infty < \epsilon$, then
\begin{align*}
&\min_{X\in K}\|M_{\qq'}(X)\|\geq\min_{X\in K}(\|(M_\qq(X)\| \\
& \qquad \qquad \qquad \qquad \qquad -\|(M_\qq-M_{\qq'})(X)\|) \geq c - \epsilon .
\end{align*}
Hence, for any $\epsilon < c$, the measurement scheme $\qq'$ determines any pure state among all states.
\end{proof}

%%%%%%%%%%%%%%%%%%%%%%
%\subsubsection*{Pure State Quantum Tomography}
%%%%%%%%%%%%%%%%%%%%%%

{\em Pure state quantum tomography.---} The most notable practical feature of  measurement schemes that determine pure states among all states is that they allow for a computationally efficient tomography of pure quantum states (see \cite{kech3}). Essentially, this is due to the fact that for every pure state $\sigma$, the unique solution to the feasibility problem 
\begin{gather*}
\text{find}\ X\\
\text{subject to}\ X\geq 0,\, M_\qq(X)=M_\qq(\sigma)
\end{gather*}
is given by $\sigma$. Note that, since $\tr{X} = \sum_j M_\qq(X)_{i,j}$, the constraints imply that any solution is a state.

%the following observation.
%\begin{proposition}
%Let $\qq$ be a PCA measurement scheme. Then, for every pure state $\sigma\in\mathcal{S}_1^d$, the unique solution to the feasibility problem 
%\begin{gather*}
%\text{find}\ X\\
%\text{subject to}\ X\geq 0,\, M_\qq(X)=M_\qq(\sigma)
%\end{gather*}
%is given by $\sigma$.
%\end{proposition}
%\begin{proof}
%Since $\tr{X'} = \sum_j M_\qq(X')_{ij}$ for all $X'\in H(d)$ and $i=1,\ldots,l$ by POVM normalization, every operator $X'\in\ker M_\qq$ is necessarily traceless. Applying this to $X'=X-\sigma$, we see that $X\in\ss^d$. Hence, the proposition is an immediate consequence of Definition \ref{def1}.
%\end{proof}
In practice, the state $\sigma$ might not be pure, but just well approximated by a pure state,  the measurement $M_\qq(\sigma)$ might be affected by systematic errors and furthermore there is statistical noise. Because of that, in a realistic scenario, one has to reconstruct $\sigma$ from the perturbed measurement data $b:=M_\qq(\sigma)+f$, where $f\in M_{lm}(\R)$ is a small error term capturing all of these sources of error. In the remainder of this section we present two convex optimization problems which allow for a recovery of any pure state $\sigma$ from the noisy measurement data $b$ provided that the measurement scheme $\qq$ determines pure states among all states.

First, consider the well-known \cite{recht2010guaranteed} semi-definite program
\begin{gather}\label{sdp2}
\begin{split}
\text{minimize}\ \text{tr}(Y)\ \ \ \ \ \ \ \ \ \ \ \ \\
\text{subject to}\ Y\geq 0,\ \|M_\qq(Y)-b\|\leq\epsilon,
\end{split}
\end{gather}
where $\epsilon>0$ is an error scale which has to be fixed in advance. %Then we get the following recovery result.
Then, as an immediate consequence of \cite[Theorem IV.1]{kech3}, we get the following recovery result (a concise proof is reported in the Supplemental Material).
\begin{theorem}[Stable Recovery I]\label{thmstab1}
Let $\epsilon>0$. There is a constant $C_\qq>0$ independent of $\epsilon$ such that for all pure states  $\sigma$ and all error terms $f\in M_{lm}(\R)$ with $\|f\|\leq\epsilon$, any minimizer $Y^*$ of \eqref{sdp2} satisfies
\begin{align*}
\|Y^*-\sigma\|\leq C_\qq\epsilon.
\end{align*}
\end{theorem}
%This result is an immediate consequence of \cite[Theorem VII.2]{kech3}. However, we give an alternative proof of this result at the end of the present paper. 

Secondly, consider the following convex program, which was also proposed in \cite{kabanava2015stable}.
\begin{gather}\label{sdp1}
\begin{split}
\ \ \ \ \ \text{minimize}\ \|M_\qq(Y)-b\|_2\\
\text{subject to}\ Y\geq 0.\ \ \ \ \ 
\end{split}
\end{gather}
Note that, different from the program \eqref{sdp2}, there is no need to guess an error scale $\epsilon$ in advance, which might be desirable from a practical point of view. The following result is then an immediate consequence of \cite[Lemma V.5]{kech3} (see also the Supplemental Material). 
\begin{theorem}[Stable Recovery II]\label{thmstab2}
Let $\epsilon>0$. There is a constant $C_\qq>0$ independent of $\epsilon$ such that for all pure states  $\sigma$ and all error terms $f\in M_{lm}(\R)$ with $\|f\|\leq\epsilon$, any minimizer $Y^*$ of \eqref{sdp1} satisfies
\begin{align*}
\|Y^*-\sigma\|\leq C_\qq\epsilon.
\end{align*}
\end{theorem}
%This result is an immediate consequence of \cite[Lemma V.5]{kech3}, but again we provide an alternative proof at the end of the present paper. 
%\begin{remark}

Note that in both Theorems \ref{thmstab1} and \ref{thmstab2}, the constant $C_\qq$ appearing in the stability bound might depend on all the parameters of $\qq$. We do not know how to estimate $C_\qq$ and hence we cannot make our stability results more explicit. Therefore, we have to rely on numerical simulations to evaluate whether the measurement schemes we constructed perform well enough in practise.
%\end{remark}

%%%%%%%%%%%%%%%%%%%%%%
%\subsubsection*{Numerical Results}
%%%%%%%%%%%%%%%%%%%%%%

{\em Numerical results.---} %In this section we present numerical results suggesting that our construction can provide ONB measurements which behave favourable in the presence of noise. 
For our simulations we choose the measurement schemes constructed from the Chebyshev polynomials of the second kind $(U_n)_{n=0}^\infty$. Moreover, we choose $\alpha=\pi/d$ and we use the Hilbert-Schmidt norm $\no{\cdot}_2$ on both $H(d)$ and $M_{lm}(\R)$. For dimensions $d=10,20,\hdots,60$, we ran the semi-definite program \eqref{sdp2} for $10^5$ times, where we sampled the pure states  and error terms $f\in M_{lm}(\R)$ with $\|f\|_2=\epsilon$ independently according to the respective Haar measures. The error scale was set to $\epsilon=10^{-4}$.

Figure \ref{fig2} shows the empiric probability density function of the reconstruction error for the dimensions $d=10,30,50$. In all cases the distribution appears to be well located, indicating a good reconstruction for most signals.
\begin{figure}[ht]
\centering
\includegraphics[width=8.5cm]{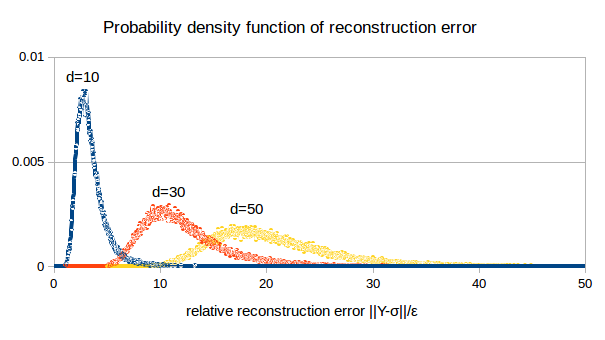}
\caption{\label{fig2} (Color online) The empiric probability density function of the relative reconstruction error $\|Y^*-\sigma\|_2/\epsilon$ for dimensions $d=10,30,50$, where $Y^*$ is the minimizer of the semi-definite program \eqref{sdp2}.}
\end{figure}

Figure \ref{fig1} shows the empiric $96\%$,$99\%$ and $99.75\%$ quantiles of the reconstruction error as well as its arithmetic mean. In the selected range of dimension the $99.75\%$ quantile error does not exceed $60\epsilon$. This suggests that for most signals the reconstruction is feasible. Furthermore, all quantiles appear to scale sublinearly with the dimension.
\begin{figure}[ht]
\centering
\includegraphics[width=8.5cm]{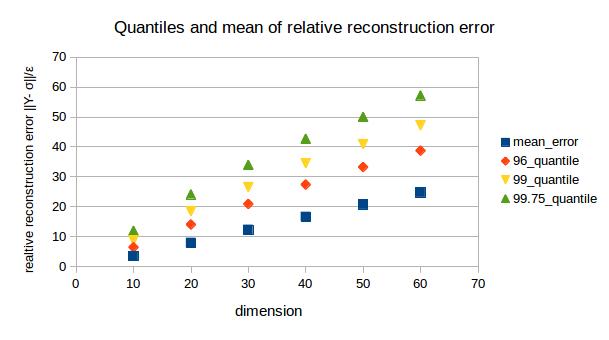}
\caption{\label{fig1} (Color online) The empiric $96\%$, $99\%$ and $99,75\%$ quantiles as well as the mean of the relative reconstruction error $\|Y^*-\sigma\|_2/\epsilon$ for dimensions $d=10,20,\hdots,60$, where $Y^*$ is the minimizer of the semi-definite program \eqref{sdp2}.}
\end{figure}

%%%%%%%%%%%%%%%%%%%%%%

%%%%%%%%%%%%
%\subsubsection*{Discussion}
%%%%%%%%%%%

%{\em Discussion.---} I think the following things should be mentioned:
%We give first explicit construction of ONB measurements that allow for pure state tomography.
%
%Say that our results come with best known number of measurements but that we cannot prove stability results as good as the one of gross et al. based on random measurements. However, numerics suggest that our measurements are pretty stable.
%
%Our measurements probably cannot be implemented in practice (I guess they can just do local things). Our construction leaves quit some freedom (we get one construction for each sequence of orthogonal polynomials). It would be interesting to explore whether this can be used to construct measurements that meet demands of experimentalists.
%
%Our results can be generalized to recovery of low rank quantum states.
%

%%%%%%%%%%
%\subsubsection*{Acknowledgments}
%%%%%%%%%%

%{\em Acknowledgments.---}

%%%%%%%%%%%%%%%%%%%%%%
\onecolumngrid

\clearpage

\section*{Supplemental Material}

For the reader's convenience, we provide the proofs of the two facts referred to in the main text as straightforward adaptations of known results. The first one is the following characterization of measurement schemes determining any pure state among all states (see \cite[Theorem 1]{CaHeScToJPA14} for a similar proof in the case of a single POVM).

\begin{proposition}\label{prop:PCA}
A measurement scheme $\qq$ determines pure states among all states if and only if every non-zero element of $\ker M_\qq$ has at least two positive eigenvalues.
\end{proposition}
\begin{proof}
The measurement scheme $\qq$ does not determine pure states among all states if and only if $M_\qq(\sigma-\varrho) = 0$ for some states $\sigma$ and $\varrho$ such that $\sigma$ is pure and $\sigma-\varrho\neq 0$. This implies that $\sigma-\varrho\in\ker M_\qq$, and $\sigma-\varrho$ has at most one positive eigenvalue by Weyl's inequalities \cite[Theorem III.2.1]{Bhatia}. Conversely, if $X\in\ker M_\qq$ is non-zero and has at most one positive eigenvalue, then it has exactly one positive eigenvalue since $\tr{X} = \sum_j M_\qq(X)_{i,j} = 0$. Hence, its positive part $X_+$ has rank $1$. Defining the states $\sigma = X_+/\tr{X_+}$ and $\varrho = (X_+-X)/\tr{X_+}$, we have that $\sigma$ is pure, $\sigma-\varrho=X/\tr{X_+}\neq 0$ and $M_\qq(\sigma-\varrho) = M_\qq(X)/\tr{X_+} = 0$.
\end{proof}

The second proof we report here is the proof of the Stable Recovery theorems, which are immediate consequences of \cite[Theorem VII.2 and Lemma V.5]{kech3}.

\begin{proof}[Proof of Theorems \ref{thmstab1} and \ref{thmstab2}.]
Note that both for the optimization problem \eqref{sdp2} and \eqref{sdp1} the minimizer $Y^*$ satisfies $\|M_\qq(Y^*)-M_\qq(\sigma)-f\|\leq\epsilon$. Hence, in both cases we find
\begin{align}
\epsilon&\geq \|M_\qq(Y^*)-M_\qq(\sigma)-f\|\nonumber \geq \|M_\qq(Y^*-\sigma)\|-\|f\|\nonumber\\
&\geq \|Y^*-\sigma\|\mathop{\inf_{X,\sigma^\prime\geq 0, \ X\neq\sigma^\prime}}_{{\rm rank}\,\sigma'=1}\frac{\|M_\qq(\sigma^\prime-X)\|}{\|\sigma^\prime-X\|}-\epsilon. \label{eqeps}
\end{align}
By Weyl's inequalities,
\begin{align*}
\left\{\frac{\sigma^\prime-X}{\no{\sigma^\prime-X}} : \ X,\sigma^\prime\geq 0, \ X\neq\sigma^\prime \text{ and } {\rm rank}\,\sigma'=1 \right\} \subseteq K := \{X'\in H(d):\lambda_2(X')\leq 0,\ \|X'\|=1\}
\end{align*}
(actually, it is easy to see that the two sets are equal). By the argument in the proof of Theorem \ref{propstab}, the set $K$ is compact. Since the measurement scheme $\qq$ determines pure states among all states, we have $M_\qq(X') \neq 0$ for all $X'\in K$ by Proposition \ref{prop:PCA}, and hence
\begin{align*}
c_\qq:=\min_{X'\in K} \|M_\qq(X')\| > 0 .
\end{align*}
This, together with \eqref{eqeps}, implies
\begin{align*}
\|Y^*-\sigma\|_2\leq \frac{2}{c_\qq}\epsilon
\end{align*}
and hence we can choose $C_\qq=2/c_\qq$.
\end{proof}
%%%%%%%%%%
%%%%%%%%%%
\end{document}